\documentclass[letterpaper,11pt]{article}
\usepackage{times}
\usepackage{a4wide}
\usepackage{latexsym}
\usepackage{amsmath}
\usepackage{amssymb}
\usepackage{ifthen}
\usepackage{mathrsfs}
\usepackage{graphics}
\usepackage{color}

\usepackage{stmaryrd}
\usepackage{amsthm}

\newcommand{\nc}{\newcommand}
\nc{\rnc}{\renewcommand} \nc{\nev}{\newenvironment}
\rnc{\subsection}{\secdef\ssa\ssb}
\nc{\ssa}[2][default]{\par\vspace{1ex}\refstepcounter{subsection}\noindent\textbf{\thesubsection.
#1. }} \nc{\ssb}[1]{\par\vspace{2ex}\noindent\textbf{#1. }}

\rnc{\subsubsection}{\secdef\sssa\sssb}
\nc{\sssa}[2][default]{\par\vspace{1ex}\refstepcounter{subsubsection}\noindent\textit{\thesubsubsection.
#1. }} \nc{\sssb}[1]{\par\vspace{1ex}\noindent\textit{#1. }}

\makeatletter
\rnc{\@seccntformat}[1]{{\normalfont\bfseries{\csname
the#1\endcsname}\hspace{1pt}.\hspace{0.4em}}}
\rnc{\section}{\@startsection
        {section}%
        {1}%
        {0mm}%
        {-\baselineskip}%
        {0.5\baselineskip}%
        {\normalfont\normalsize\bfseries\centering}%
}
\renewcommand{\@makecaption}[2]{\begin{center}#1. #2\end{center}}
\makeatother

\newtheorem{theo}{Theorem}[section]
\newtheorem{lem}[theo]{Lemma}
\newtheorem{cor}[theo]{Corollary}
\newtheorem{prop}[theo]{Proposition}

\theoremstyle{definition}
\newtheorem{defn}[theo]{Definition}
\newtheorem{rem}[theo]{Remark}

\rnc{\proof}[1][{}]{\smallskip\noindent\textit{Proof #1: }}
\nc{\proofend}{\hfill$\Box$\vspace{\topsep}\par}

\rnc{\labelenumi}{(\arabic{enumi})} \rnc{\labelitemi}{\text{--}}
\rnc{\phi}{\varphi} \rnc{\epsilon}{\varepsilon}
\nc{\bigmid}{\;\big|\;} \nc{\Bigmid}{\;\Big|\;}
\rnc{\max}{\textup{max}} \rnc{\min}{\textup{min}}
\rnc{\log}{\textup{log}\;}

\newlength{\probwidth}
\nc{\prob}[3][9]{
\begin{center}
  \normalfont\fbox{
   \begin{tabular}[t]{
     rp{#1cm}}\textit{Instance:}&#2. \\
     \textit{Problem:}&#3
   \end{tabular}}
\end{center}}

\nc{\pprob}[4][9]{
\begin{center}
   \normalfont\fbox{
    \begin{tabular}[t]{
     rp{#1cm}}\textit{Instance:}&#2. \\
     \textit{Parameter:}&#3. \\
     \textit{Problem:}&#4
   \end{tabular}}
\end{center}}

\nc{\nprob}[4][9]{
\begin{center}
  \normalfont\fbox{

\addtolength{\probwidth}{#1cm}\parbox{\probwidth}{\textsc{#2}\\\hspace*{1.5em}
     \begin{tabular}[t]{
      rp{#1cm}}\textit{Instance:}&#3. \\
      \textit{Problem:}&#4
     \end{tabular}}}
\end{center}}

\nc{\npprob}[5][9]{
\begin{center}
  \normalfont\fbox{

\addtolength{\probwidth}{#1cm}\parbox{\probwidth}{\textsc{#2}\\\hspace*{1.5em}
    \begin{tabular}[t]{
     rp{#1cm}}\textit{Instance:}&#3. \\
     \textit{Parameter:}&#4. \\
     \textit{Problem:}&#5
    \end{tabular}}}
\end{center}}

\nc{\nppxrob}[5][9]{ \normalfont\fbox{

\addtolength{\probwidth}{#1cm}\parbox{\probwidth}{\textsc{#2}\\\hspace*{1.5em}
   \begin{tabular}[t]{
    rp{#1cm}}\textit{Instance:}&#3. \\
    \textit{Parameter:}&#4. \\
    \textit{Problem:}&#5
   \end{tabular}}}}

\nc{\nppprob}[5][4]{
\begin{center}
  \normalfont\fbox{

\addtolength{\probwidth}{#1cm}\parbox{\probwidth}{\textsc{#2}\\\hspace*{1.5em}
    \begin{tabular}[t]{
     rp{#1cm}}\textit{Instance:}&#3. \\
     \textit{Parameter:}&#4. \\
     \textit{Problem:}&#5
    \end{tabular}}}
\end{center}}

\nc{\noptprob}[6][9]{
\begin{center}
  \normalfont\fbox{

\addtolength{\probwidth}{#1cm}\parbox{\probwidth}{\textsc{#2}\\\hspace*{1.5em}
    \begin{tabular}[t]{
     rp{#1cm}}\textit{Instance:}&#3. \\
     \textit{Solution:}&#4. \\
     \textit{Cost:}&#5. \\
     \textit{Goal:}&#6.
    \end{tabular}}}
\end{center}}

\nc{\FOR}{\textbf{for}}
\nc{\FORALL}{\textbf{for all}}
\nc{\TO}{\textbf{to}}
\nc{\DO}{\textbf{do}}
\nc{\OD}{\textbf{od}}
\nc{\IF}{\textbf{if}}
\nc{\FI}{\textbf{fi}}
\nc{\THEN}{\textbf{then}}
\nc{\ELSE}{\textbf{else}}
\nc{\WHILE}{\textbf{while}}
\nc{\REPEAT}{\textbf{repeat}}
\nc{\UNTIL}{\textbf{until}}
\nc{\OR}{\textbf{or}}
\nc{\AND}{\textbf{and}}
\nc{\PRINT}{\textbf{print}}

\nc{\im}[1]{\item\hspace{#1cm}}
\nev{algorithm}{\begin{enumerate}\rnc{\labelenumi}{\textit{\small \arabic{enumi}.}}\rnc{\itemsep}{0ex}}{\end{enumerate}}

\nc{\fpcl}[1]{\left[#1\right]_{\text{\upshape fp}}}
\nc{\pr}{\le^{\text{\normalfont fp}}_m} \nc{\FPT}{\textup{FPT}}
\nc{\EPT}{\textup{EPT}} \nc{\SUBEPT}{\textup{SUBEPT}}

\nc{\fpt}{\textup{fpt}} \nc{\fptT}{\textup{fpt-T}}
\nc{\W}[1]{\text{$\textup{W}[#1]$}}
\nc{\M}[1]{\text{$\textup{M}[#1]$}}
\nc{\MS}[2]{\text{$\textup{M}^{#1}[#2]$}}
\nc{\MINI}[1]{\mbox{\small \rm MINI[$#1$]}}
\nc{\WP}{\textup{W[P]}} \nc{\AWP}{\textup{AW[P]}}

\rnc{\S}[1]{\text{$\textup{S}[#1]$}} \nc{\SP}{\textup{S[P]}}
\nc{\MP}{\textup{M[P]}}

\nc{\PTIME}{\textup{PTIME}} \nc{\APTIME}{\textup{APTIME}}
\nc{\PSPACE}{\textup{PSPACE}} \nc{\NP}{\textup{NP}}

\nc{\DTIME}{\textup{DTIME}}

\nc{\se}{\subseteq} \nc{\re}{\rightarrow}

\nc{\LOEFF}[1]{{o}^{\rm eff}(#1)}

\nc{\PNPTC}{\mbox{$\textup{P}[{\textsc{tc}}]\ne
\textup{NP}[{\textsc{tc}}]$}} \nc{\str}[1]{\ensuremath{\mathcal #1}}
\nc{\cls}[1]{\ensuremath{\mathbf #1}}

\nc{\algo}[1]{{\mathbb #1}}

\nc{\NAT}{{\mathbb N}}

\nc{\VC}{\textsc{Vertex-Cover}} \nc{\IS}{\textsc{Independent-Set}}
\nc{\Cli}{\textsc{Clique}} \nc{\DS}{\textsc{Dominating-Set}}
\nc{\TSAT}{\textsc{3-Sat}}

\nc{\CNF}{\textup{CNF}} \nc{\WSAT}{\textsc{WSat}}
\nc{\AWSAT}{\textsc{AWSat}} \nc{\SAT}{\textsc{Sat}}
\nc{\ASAT}{\textsc{ASat}} \nc{\CIRC}{\textsc{Circ}}
\nc{\PROP}{\textsc{Prop}}

\nc{\nva}{\textup{nv}} \nc{\ncl}{\textup{nc}}

\nc{\ETH}{\textup{ETH}}
\nc{\Wone}{\textup{W[1]}}

\nc{\serf}{\textup{serf}} \nc{\serfT}{\textup{serf-T}}
\nc{\var}{\textup{var}}

\nc{\NUXP}{\textup{XP}_{\rm nu}} \nc{\XP}{\textup{XP}}
\nc{\SNP}{\textup{SNP}}

\nc{\ept}{{\rm ept}}

\nc{\E}{\textup{E}}

\nc{\HALT}{\textsc{Halt}}

\nc{\TM}{\textsc{TM}} \nc{\TMBA}{\textsc{TMBA}}

\nc{\pCLIQUE}{\textup{$p$-$\textsc{Clique}$}}

\nc{\ceil}[1]{\left\lceil#1\right\rceil}
\nc{\floor}[1]{\left\lfloor#1\right\rfloor}

\nc{\bende}{\eqno$\Box$} \nc{\benda}{\tag*{$\Box$}}

\nc{\pa}{\kappa}

\nc{\co}{\textup{co-}}

\rnc{\L}{\textup{LOGSPACE}}

\nc{\NL}{\textup{NLOGSPACE}}

\rnc{\P}{\textup{P}}

\rnc{\angle}[1]{\langle #1\rangle}

\nc{\rand}[1]{\marginpar{\raggedright\footnotesize #1}}
\nc{\yrand}[1]{\rand{\textbf{Y: }#1}}
\nc{\brand}[1]{\rand{\textbf{B: }#1}}

\nc{\Ppoly}{\textup{P}/\text{\small poly}}

\nc{\f}{\mathbf f}

\nc{\s}{\mathbf s}

\nc{\tim}{\textup{time}}

\nc{\sat}{\textup{sat}}
\nc{\rank}{\textup{rank}}
\nc{\PC}{\mathbf{PC}}
\nc{\bin}{{\rm in}}
\nc{\bout}{{\rm out}}
\nc{\Mix}{\textsc{Mix}}
\nc{\MIX}{\mathbf{MIX}}

\nc{\tdeg}{\textup{tdeg}}
\nc{\lspan}{\textup{span}}

\nc{\ma}[1]{\mathbb #1}
\nc{\bet}[1]{\| #1\|}
\nc{\EFM}[2]{(#1,#2)}
\nc{\EF}{\text{Ehren\-feucht\--Fra\"i\-ss\'e}}
\nc{\AF}{\textup{AF}}
\nc{\supp}{\textup{supp}}
\nc{\ar}{\textup{ar}}
\nc{\pol}{\textup{pol-}}
\nc{\equivm}{\equiv_{\LFP_m}}
\nc{\equivfo}{\equiv_{\FO_m}}
\nc{\ERRE}{Erd\H{o}s-R\'enyi}

\nc{\PCC}{\textup{PCC}}

\nc{\TCOL}{\textsc{3-Col}}

\nc{\GI}{\textup{GI}}

\nc{\EX}{\textup{E}}

\nc{\Var}{\textup{Var}}

\nc{\size}{\textup{size}}

\nc{\DNF}{\textup{DNF}}

\nc{\n}{\tilde n}

\nc{\dotcup}{\;\dot\cup\;}

\nc{\ds}{\gamma}

\nc{\mds}{\ensuremath{\textsc{Min-Dominating-Set}}}
\nc{\pds}{\ensuremath{p\textsc{-Dominating-Set}}}
\nc{\pclique}{\ensuremath{p\textsc{-Clique}}}
\nc{\mcs}{\ensuremath{\textsc{Monotone-Circuit-Satisfiability}}}

\nc{\sol}{\textup{sol}}
\nc{\cost}{\textup{cost}}
\nc{\goal}{\textup{goal}}

\nc{\pow}{\text{Pow}}

\begin{document}

\title{The Constant Inapproximability of the Parameterized Dominating Set Problem}
\author{Yijia Chen\\\normalsize School of Computer Science\\
\normalsize Fudan University\\
\normalsize yijiachen@fudan.edu.cn\\
\and
Bingkai Lin\\\normalsize Department of Computer Science\\
\normalsize University of Tokyo\\
\normalsize lin@is.s.u-tokyo.ac.jp}

\date{}

\maketitle

\begin{abstract}
We prove that there is no \fpt-algorithm that can approximate the
dominating set problem with any constant ratio, unless $\FPT= \W 1$. Our
hardness reduction is built on the second author's recent $\W 1$-hardness
proof of the biclique problem~\cite{lin15}. This yields, among other
things, a proof without the PCP machinery that the classical dominating
set problem has no polynomial time constant approximation under the
exponential time hypothesis.
\end{abstract}

{

\rnc{\c}{\textbf{\emph{c}}} \rnc{\a}{\textbf{\emph{a}}}
\rnc{\b}{\textbf{\emph{b}}} \rnc{\i}{\textbf{\emph{i}}}
\rnc{\j}{\textbf{\emph{j}}} \rnc{\u}{\textbf{\emph{u}}}
\rnc{\v}{\textbf{\emph{v}}}

\section{Introduction}
The dominating set problem, or equivalently the set cover problem, was among
the first problems proved to be \NP-hard~\cite{kar72}. Moreover, it has been
long known that the greedy algorithm achieves an approximation ratio
$\approx \ln n$~\cite{joh74,ste74,lov75,cha79,sla97}. And after a sequence
of papers (e.g.~\cite{lunyan94,razsaf97,fei98,alomos06,dinste14}), this is
proved to be best possible. In particular, Raz and Safra~\cite{razsaf97}
showed that the dominating set problem cannot be approximated with ratio
$c\cdot \log n$ for some constance $c\in \mathbb N$ unless $\P=
\NP$~\cite{razsaf97}. Under a stronger assumption $\NP\not\subseteq
\DTIME\left(n^{O(\log\log n)}\right)$ Feige proved that no approximation
within $(1- \varepsilon)\ln n$ is feasible~\cite{fei98}. Finally Dinur and
Steuer established the same lower bound assuming only $\P\ne
\NP$~\cite{dinste14}. However, it is important to note that the
approximation ratio $\ln n$ is measured in terms of the size of an input
graph $G$, instead of $\ds(G)$, i.e., the size of its minimum dominating
set. As a matter of fact, the standard examples for showing the $\Theta(\log
n)$ greedy lower bound have constant-size dominating sets. Thus, the size of
the greedy solutions cannot be bounded by any function of $\ds(G)$. So the
question arises whether there is an approximation algorithm $\mathbb A$ that
always outputs a dominating set whose size can be bounded by
$\rho(\ds(G))\cdot \ds(G)$, where the function $\rho:\mathbb N\to\mathbb N$
is known as the approximation ratio of $\mathbb A$. The constructions
in~\cite{fei98,alomos06} indeed show that we can rule out $\rho(x)\le \ln
x$. To the best of our knowledge, it is not known whether this bound is
tight. For instance, it is still conceivable that there is a polynomial time
algorithm that always outputs a dominating set of size at most
$2^{2^{\ds(G)}}$.

Other than looking for approximate solutions, parameterized
complexity~\cite{dowfel99, flugro06, nie06, dowfel13, cygfom15} approaches
the dominating set problem from a different perspective. With the
expectation that in practice we are mostly interested in graphs with
relatively small dominating sets, algorithms of running time
$2^{\ds(G)}\cdot |G|^{O(1)}$ can still be considered efficient.
Unfortunately, it turns out that the parameterized dominating set problem is
complete for the second level of the so-called W-hierarchy~\cite{dowfel95},
and thus fixed-parameter intractable unless $\FPT= \W 2$. So one natural
follow-up question is whether the problem can be approximated in fpt-time.
More precisely, we aim for an algorithm with running time $f(\gamma(G))\cdot
|G|^{O(1)}$ which always outputs a dominating set of size at most
$\rho(\ds(G))\cdot \ds(G)$. Here, $f:\mathbb N\to \mathbb N$ is an arbitrary
computable function. The study of parameterized approximability was
initiated in~\cite{caihua10, chegro07, dowfel06}. Compared to the classical
polynomial time approximation, the area is still in its very early stage
with few known positive and even less negative results.

\subsection*{Our results}
We prove that any constant-approximation of the parameterized dominating set
problem is \W 1-hard.
\begin{theo}\label{thm:main1}
For any constant $c\in \mathbb N$ there is \emph{no} \fpt-algorithm $\mathbb
A$ such that on every input graph $G$ the algorithm $\mathbb A$ outputs a
dominating set of size at most $c\cdot \ds(G)$, unless $\FPT= \W 1$ (which
implies that the exponential time hypothesis (\ETH) fails).
\end{theo}

In the above statement, clearly we can replace ``\fpt-algorithm'' by
``polynomial time algorithm,'' thereby obtaining the classical
constant-inapproximability of the dominating set problem. But let us mention
that our result is not comparable to the classical version, even if we
restrict ourselves to polynomial time tractability. The assumption $\FPT\ne
\W 1$ or \ETH\ is apparently much stronger than $\P\ne \NP$, and in fact
\ETH\ implies $\NP\not\subseteq \DTIME\left(n^{O(\log\log n)}\right)$ used
in aforementioned Feige's result. But on the other hand, our lower bound
applies even in case that we know in advance that a given graph has no large
dominating set.

\begin{cor}\label{cor:main1}
Let $\beta:\mathbb N\to \mathbb N$ be a nondecreasing and unbounded
computable function. Consider the following promise problem.
\noptprob[6.5]{$\mds_\beta$}{A graph $G=(V,E)$ with $\ds(G)\le
\beta(|V|)$}{A dominating set $D$ of $G$}{$|D|$}{$\min$}
Then there is no polynomial time constant approximation algorithm for
$\mds_{\beta}$, unless $\FPT= \W 1$.
\end{cor}

\medskip
The proof of Theorem~\ref{thm:main1} is crucially built on a recent result
of the second author~\cite{lin15} which shows that the parameterized
biclique problem is $\W 1$-hard. We exploit the gap created in its hardness
reduction (see Section~\ref{subsec:biclique} for more details). In the known
proofs of the classical inapproximability of the dominating set problem, one
always needs the PCP theorem in order to have such a gap, which makes those
proofs highly non-elementary. More importantly, it can be verified that
reductions based on the PCP theorem produce instances with optimal solutions
of relatively large size, e.g., a graph $G=(V,E)$ with $\ds(G)\ge
|V|^{\Theta(1)}$. This is inevitable, since otherwise we might be able to
solve every \NP-hard problem in subexponential time. As an example, if it is
possible to reduce an \NP-hard problem to the approximation of
$\textsc{Min-Dominating-Set}_\beta$ for $\beta(n)= \log \log \log n$, then
by brute-force searching for a minimum dominating set, we are able to solve
the problem in time $n^{O(\log\log\log n)}$. It implies $\NP\subseteq
\DTIME\left(n^{O(\log\log\log n)}\right)$. Because of this,
Corollary~\ref{cor:main1}, and hence also Theorem~\ref{thm:main1}, is
unlikely provable following the traditional approach.

\medskip
Using a result of Chen et.al.~\cite{chehua04} the lower bound in
Theorem~\ref{thm:main1} can be further sharpened.
\begin{theo}\label{thm:main2}
Assume \ETH\ holds.
Then there is \emph{no} \fpt-algorithm which on every input graph $G$
outputs a dominating set of size at most $\sqrt[4+\varepsilon]{\log
(\ds(G))} \cdot \ds(G)$ for every $0<\varepsilon<1$.
\end{theo}

\subsection*{Related work}
The existing literature on the dominating set problem is vast. The most
relevant to our work is the classical approximation upper and lower bounds
as explained in the beginning. But as far as the parameterized setting is
concerned, what was known is rather limited.

Downey et. al proved that there is no \emph{additive} approximation of the
the parameterized dominating set problem~\cite{dowfel08}. In the same paper,
they also showed that the independent dominating set problem has no \fpt\
approximation with any approximation ratio. Recall that an independent
dominating set is a dominating set which is an independent set at the same
time. With this additional requirement, the problem is no longer
\emph{monotone}, i.e., a superset of a solution is not necessarily a
solution. Thus it is unclear how to reduce the independent dominating set
problem to the dominating set problem by an approximation-preserving
reduction.

In~\cite{chihajkor13, hajkhakor13} it is proved under \ETH\ that
there is no $c\sqrt{\log \ds(G)}$-approximation algorithm for the dominating
set problem\footnote{The papers actually address the set cover problem,
which is equivalent to the dominating set problem as mentioned in the
beginning.} with running time
$2^{O(\ds(G)^{(\log \ds(G))^d})}|G|^{O(1)}$, where $c$ and $d$ are some
appropriate constants. With the additional \emph{Projection Game Conjecture}
due to~\cite{mos12} and some of its further strengthening, the authors of
\cite{chihajkor13, hajkhakor13} are able to even rule out
$\ds(G)^c$-approximation algorithms with running time almost doubly
exponential in terms of $\ds(G)$.
Clearly, these lower bounds are against far better approximation ratio than
those of Theorem~\ref{thm:main1} and Theorem~\ref{thm:main2}, while the drawback is that the dependence of the running time on $\ds(G)$ is not an arbitrary computable function.

The dominating set problem can be understood as a special case of the
weighted satisfiability problem of \CNF-formulas, in which all literals are
positive. The weighted satisfiability problems for various fragments of
propositional logic formulas, or more generally circuits, play very
important roles in parameterized complexity. In particular, they are
complete for the W-classes. In~\cite{chegro07} it is shown that they have no
\fpt\ approximation of any possible ratio, again by using the
non-monotoncity of the problems. Marx strengthened this result significantly
in~\cite{mar13} by proving that the weighted satisfiability problem is not
\fpt\ approximable for circuits of depth 4 without negation gates, unless
$\FPT = \W 2$. Our result can be viewed as an attempt to improve Marx's
result to depth-2 circuits, although at the moment we are only able to rule
out \fpt\ approximations with constant ratio.

\subsection*{Organization of the paper} We fix our notations in
Section~\ref{sec:pre}. In the same section we also explain the result
in~\cite{lin15} key to our proof. To help readability, we first prove that
the dominating set problem is not \fpt\ approximable with ratio smaller than
$3/2$ in Section~\ref{sec:32}. In the case of the clique problem, once we
have inapproximability for a particular constant ratio, it can be easily
improved to any constant by gap-amplification via graph products. But
dominating sets for general graph products are notoriously
hard to understand (see e.g.~\cite{klazma96}). So to prove
Theorem~\ref{thm:main1}, Section~\ref{sec:main} presents a modified
reduction which contains a tailor-made graph product.
Section~\ref{sec:consq} discusses some consequences of our results. We
conclude in~Section~\ref{sec:con}.

\section{Preliminaries}\label{sec:pre}
We assume familiarity with basic combinatorial optimizations and
parameterized complexity, so we only introduce those notions and notations
central to our purpose. The reader is referred to the standard textbooks
(e.g., ~\cite{auscre99} and \cite{dowfel99,flugro06}) for further
background.

\medskip
$\mathbb N$ and $\mathbb N^+$ denote the sets of natural numbers (that is,
nonnegative integers) and positive integers, respectively. For every $n\in
\mathbb N$ we let $[n]:= \{1, \ldots, n\}$. $\mathbb R$ is the set of real
numbers, and $\mathbb R_{\ge 1}:= \big\{r\in \mathbb R\bigmid r\ge 1\big\}$.
For a function $f: A\to B$ we can extend it to sets and vectors by defining
$f(S):= \{f(x)\mid x\in S\}$ and $f(\v):=\big(f(v_1), f(v_2), \cdots,
f(v_k)\big)$, where $S\subseteq A$ and $\v = (v_1, v_2, \cdots, v_k)\in A^k$
for some $k\in \mathbb{N}^+$.

\medskip
Graphs $G= (V,E)$ are always simple, i.e., undirected and without loops and
multiple edges. Here, $V$ is the vertex set and $E$ the edge set,
respectively. The \emph{size} of $G$ is $|G|:= |V|+|E|$. A subset
$D\subseteq V$ is a \emph{dominating set} of $G$, if for every $v\in V$
either $v\in D$ or there exists a $u\in D$ with $\{u,v\}\in E$. In the
second case, we might say that $v$ is dominated by $u$, and this can be
easily generalized to $v$ dominated by a set of vertices. The
\emph{domination number} $\ds(G)$ of $G$ is the size of a smallest
dominating set. The classical \emph{minimum dominating set problem} is to
find such a dominating set:
\noptprob[5.5]{\mds}{A graph $G=(V,E)$}{A dominating set $D$ of
$G$}{$|D|$}{$\min$}
The decision version of $\mds$ has an additional input $k\in \mathbb N$.
Thereby, we ask for a dominating set of size at most $k$ instead of
$\ds(G)$. But it is well known that two versions can be reduced to each
other in polynomial time. In parameterized complexity, we view the input $k$
as the parameter and thus obtain the standard parameterization of
$\textsc{Min-Dominating-Set}$:
\npprob[8.7]{\pds}{A graph $G$ and $k\in \mathbb N$}{$k$}{Decide whether $G$
has a dominating set of size at most $k$.}
As mentioned in the Introduction, \pds\ is complete for the parameterized
complexity class $\W 2$, the second level of the W-hierarchy. We will need
another important parameterized problem, the \emph{parameterized clique
problem}
\npprob[8.7]{\pclique}{A graph $G$ and $k\in \mathbb N$}{$k$}{Decide whether
$G$ has a clique of size at most $k$.}
which is complete for $\W 1$. Recall that a subset $S\subseteq V$ is a
\emph{clique} in $G= (V,E)$, if for every $u,v \in S$ we have either $u=v$
or $\{u,v\}\in E$.

Those W-classes are defined by weighted satisfiability problems for
propositional formulas and circuits. As they will be used only in
Section~\ref{sec:consq}, we postpone their definition until then.

\subsection*{Parameterized approximability} We follow the general framework
of~\cite{chegro07}. However,  to lessen the notational burden we restrict
our attention to the approximation of the dominating set problem.

\begin{defn}\label{def:paraapp}
Let $\rho:\mathbb N\to \mathbb R_{\ge 1}$. An algorithm $\mathbb A$ is a
\emph{parameterized approximation algorithm} for \pds\ with approximation
ratio $\rho$ if for every graph $G$ and $k\in \mathbb N$ with $\ds(G)\le k$
the algorithm $\mathbb A$ computes a dominating set $D$ of $G$ such that
\[
|D|\le \rho(k)\cdot k.
\]
If the running time of $\mathbb A$ is bounded by $f(k)\cdot |G|^{O(1)}$
where $f:\mathbb N\to \mathbb N$ is computable, then $\mathbb A$ is an \fpt\
approximation algorithm.
\end{defn}

One might also define parameterized approximation directly for \mds\ by
taking $\ds(G)$ as the parameter. The next result shows that essentially
this leads to the same notion.

\begin{prop}[{\cite[Proposition~5]{chegro07}}]
Let $\rho:\mathbb N\to \mathbb R_{\ge 1}$ be a function such that
$\rho(k)\cdot k$ is nondecreasing. Then the following are equivalent.
\begin{enumerate}
\item \pds\ has an \fpt\ approximation algorithm with approximation ratio
    $\rho$.

\item There exists a computable function $g:\mathbb N\to \mathbb N$ and an
    algorithm $A$ that on every graph $G$ computes a dominating set $D$ of $G$ with $|D|\le
    \rho(\ds(G))\cdot \ds(G)$ in time $g(\ds(G))\cdot |G|^{O(1)}$.
\end{enumerate}
\end{prop}

\subsection*{The Color-Coding}
\begin{lem}[\cite{aloyus95}]\label{lem:cc}
For every $n,k\in \mathbb N$ there is a family $\Lambda_{n,k}$ of polynomial
time computable functions from $[n]$ to $[k]$ such that for every
$k$-element subset $X$ of $[n]$, there is an $h\in \Lambda_{n,k}$ such that
$h$ is injective on $X$. Moreover, $\Lambda_{n,k}$ can be computed in time
$2^{O(k)}\cdot n^{O(1)}$.
\end{lem}

\subsection{The \W 1-hardness reduction of the parameterized biclique
problem}\label{subsec:biclique} Our starting point is the following theorem
proved in~\cite{lin15} which states that, on input a bipartite graph, it is
$\Wone$-hard to distinguish whether there exist $k$ vertices with large
number of common neighbors or every $k$-vertex set has small number of
common neighbors.

\begin{theo}[{\cite[Theorem~1.3]{lin15}}]\label{thm:gapreduction}
There is a polynomial time algorithm $\mathbb A$ such that for every graph
$G$ with $n$ vertices and $k\in \mathbb N$ with $\ceil{n^{\frac{6}{k+6}}}>
(k+6)!$ and $6\mid k+1$ the algorithm $\mathbb A$ constructs a bipartite
graph $H=(A \dotcup B, E)$ satisfying:
\begin{enumerate}
\item if $G$ contains a clique of size $k$, i.e., $K_k\subseteq G$, then
    there are $s$ vertices in $A$ with at least $\ceil
    {n^{\frac{6}{k+1}}}$ common neighbors in $B$;

\item otherwise $K_k\nsubseteq G$, every $s$ vertices in $A$ have at most
    $(k+1)!$ common neighbors in $B$,
\end{enumerate}
where $s=\binom{k}{2}$.
\end{theo}

In our reductions from $\pclique$ to $\pds$, we use the following procedure
to ensure that the instance $(G,k)$ of $\pclique$ satisfies $6\mid k+1$.

\bigskip
\noindent\textbf{Preprocessing.} On input a graph $G$ and $k\in \mathbb
N^+$, if $6$ does not divide $k+1$, let $k'$ be the minimum integer such
that $k'\ge k$ and $6\mid k'+1$. We construct a new graph $G'$ by adding a
clique with $k'-k$ vertices into $G$ and making every vertex of this clique
adjacent to other vertices in $G$. It is easy to see that $k'\le k+5$, and
$G$ contains a $k$-clique if and only if $G'$ contains a $k'$-clique. Then
we proceed with $G\gets G'$ and $k\gets k'$.

\section{The Case $\rho< 3/2$}\label{sec:32}

As the first illustration of how to use the gap created in
Theorem~\ref{thm:gapreduction}, we show in this section that $\pds$ cannot
be fpt approximated within ratio $< 3/2$. This serves as a stepping stone to
the general constant-inapproximability of the problem.

\begin{theo}
Let $\rho< 3/2$. Then there is no \fpt\ approximation of the parameterized
dominating set problem achieving  ratio $\rho$ unless $\FPT=\W 1$.
\end{theo}

\proof We fix some $\varepsilon, \delta\in \mathbb R$ with $0< \varepsilon<
1$, $0< \delta< 1/2$,  and
\begin{equation}\label{eq:rhoepsdelta}
\frac{3/2- \delta}{1+ \varepsilon} > \rho.
\end{equation}

Let $G$ be a graph with $n$ vertices and $k\in \mathbb N$ a parameter. We
set $s:= \binom{k}{2}$,
\begin{eqnarray*}
d:= \ceil{\frac{s}{\varepsilon}}^{2s},
 & \text{and} &
t:= \ceil{\left(\frac{1}{2}- \delta\right)\cdot d^{1-1/2s}}.
\end{eqnarray*}
As a consequence, when $k$ and $n$ are sufficiently large, we have
\begin{equation}\label{eq:at}
s t< \varepsilon d,
 \quad \left(\frac{1}{2}- \delta\right)\cdot \frac{d}{t}\le \sqrt[2s]{d},
 \quad (k+1)!<  2 \delta \sqrt{d}-1,
 \quad \text{and} \quad d\le \lceil n^{\frac{6}{k+1}}\rceil.
\end{equation}

By Theorem~\ref{thm:gapreduction} (and the preprocessing) we can compute in
\fpt-time a bipartite graph $H_0=(A_0 \dotcup B_0, E_0)$ such that:
\begin{enumerate}
\item[-] if $K_k\subseteq G$, then there are $s$ vertices in $A_0$ with
    $d$ common neighbors in $B_0$;

\item[-] if $K_k\nsubseteq G$, then every $s$ vertices in $A_0$ have at
    most $(k+1)!$ common neighbors in $B_0$.
\end{enumerate}
Then using the color-coding in Lemma~\ref{lem:cc}, again in \fpt-time, we
construct two function families $\Lambda_A:= \Lambda_{|A_0|, s}$ and
$\Lambda_B:= \Lambda_{|B_0|, d}$ such that
\begin{enumerate}
\item[-] for every $s$-element subset $X\subseteq A_0$ there is an $h\in
    \Lambda_A$ with $h(X)= [s]$;

\item[-] for every $d$-element subset $Y\subseteq B_0$ there is an $h\in
    \Lambda_B$ with $h(Y)= [d]$.
\end{enumerate}
Define the bipartite graph $H= \big(A(H) \dotcup B(H), E(H)\big)$
by
\begin{align*}
A(H) &:= A_0\times \Lambda_A\times \Lambda_B,
 \qquad B(H):= B_0\times  \Lambda_A\times \Lambda_B\\
E(H) &:= \Big\{\big\{(u,h_1,h_2), (v, h_1,h_2)\big\}
 \Bigmid \text{$u\in A_0$, $v\in B_0$, $h_1\in \Lambda_A$, $h_2\in \Lambda_B$,
  and $\{u,v\}\in E_0$ }\Big\}.
\end{align*}
Moreover, define two colorings $\alpha: A(H)\to [s]$ and $\beta: B(H)\to
[d]$ by
\begin{eqnarray*}
\alpha(u,h_1,h_2):= h_1(u) & \text{and} & \beta(v,h_1,h_2):= h_2(v).
\end{eqnarray*}
It is straightforward to verify that
\begin{enumerate}
\item[(H1)] if  $K_k\subseteq G$, then there are $s$ vertices of distinct
    $\alpha$-colors in $A(H)$ with $d$ common neighbors of distinct
    $\beta$-colors in $B(H)$;

\item[(H2)] if  $K_k\nsubseteq G$, then every $s$ vertices in $A(H)$ have
    at most $(k+1)!$ common neighbors in $B(H)$.
\end{enumerate}

%

\medskip
Now from $H$, $\alpha$, and $\beta$ we construct a new graph $G'=
\big(V(G'), E(G')\big)$ as
follows. First, its vertex set is defined by
\[
V(G'):= B(H) \dotcup \big\{x_i, y_i \bigmid i\in [d]\big\}
 \dotcup C \dotcup W,
\]
where
\begin{eqnarray*}
C:=  A(H) \times [t] & \text{and} &
W:= \left\{w_{b,j,i} \Bigmid b\in B(H), i\in [t], j\in [s]\right\}.
\end{eqnarray*}
Moreover, $G'$ contains the following types of edges.
\begin{enumerate}
\item[(E1)] $\{b,b'\}\in E(G')$ with $b,b'\in B(H)$, $b\ne b'$, and
    $\beta(b)= \beta(b')$ \big(i.e., all vertices in $B(H)$ with the same
    color under $\beta$ form a clique in $G'$\big).

\item[(E2)] Let $b\in B(H)$ and $c:= \beta(b)$. Then $\{x_c, b\}, \{y_c,
    b\}\in E(G')$.

\item[(E3)] Let $b,b'\in B(H)$ with $\beta(b)= \beta(b')$ and $b\ne b'$.
    Then $\big\{w_{b,j,i}, b'\big\}\in E(G')$ for every $i\in [t]$ and
    $j\in [s]$.

\item[(E4)] $\big\{(a,i), w_{b,j,i}\big\}\in E(G')$ for every $\{a,b\}\in
    E(H)$, $j= \alpha(a)$ and $i\in [t]$.

\item[(E5)] Let $a, a'\in A(H)$ with $a\ne a'$ and $i\in [t]$. Then
    $\big\{(a,i), (a',i)\big\}\in E(G')$.
\end{enumerate}

\medskip
To ease presentation, for every $c\in [d]$ we set
\[
B_c:= \big\{b\in B(H)\bigmid \beta(b)= c\big\} \cup \{x_c, y_c\}.
\]

\medskip
\noindent \textit{Claim 1.} If $D$ is a dominating set of $G'$, then $D\cap
B_c\ne \emptyset$ for every $c\in [d]$.

\medskip
\noindent \textit{Proof of the claim.} We observe that every $x_c$ is only
adjacent to vertices in $B_c$. \hfill$\dashv$

\bigskip \noindent \textit{Claim 2.} If $G$ contains a $k$-clique, then
$\ds(G')< (1+\varepsilon) d$.

\medskip
\noindent \textit{Proof of the claim.} By (H1) the bipartite graph $H$ has a
$K_{s,d}$ biclique $K$ with $\alpha(A(H)\cap K)= [s]$ and $\beta(B(H)\cap
K)=[d]$. It is then easy to verify that
\[
\big(B(H)\cap K\big) \dotcup \big((A(H)\cap K)\times [t]\big)
\]
is a dominating set of $G'$, whose size is $d+ s\cdot t< (1+\varepsilon) d$
by~\eqref{eq:at}. \hfill$\dashv$

\bigskip \noindent \textit{Claim 3.} If $G$ contains no $k$-clique, then
every $s$-vertex set of $A(H)$ has at most $(k+1)!< 2 \delta \sqrt{d}-1$
common neighbors in $B(H)$.

\medskip
\noindent \textit{Claim 4.} If $G$ contains no $k$-clique, then
\[
\ds(G')> \left(\frac{3}{2}- \delta\right)\cdot d.
\]

\medskip
\noindent \textit{Proof of the claim.} Let $D$ be a dominating set of $G'$.
By Claim~1 we have $D\cap B_c\ne \emptyset$ for every $c\in [d]$. Define
\[
e:= \Big|\big\{c\in [d] \bigmid |D\cap B_c|\ge 2\big\}\Big|.
\]
If $e> (1/2- \delta) \cdot d$ then $|D|> d+ e> (3/2- \delta)\cdot d$ and we
are done.

\medskip
So let us consider $e\le (1/2- \delta) \cdot d$ and without loss of
generality $|D\cap B_c| = 1$ for every $c\le (1/2+ \delta) \cdot d$. Fix
such a $c$ and assume $D\cap B_c = \{b_c\}$. Recall $x_c, y_c\in V(G')$ are
not adjacent to any vertex outside $B_c$, and there is no edge between them,
thus $b_c\in B_c\setminus \{x_c, y_c\}= \big\{b\in B(H)\bigmid \alpha(b)=
c\big\}$. Let
\[
W_1:= \left\{w_{b_c, j,i} \Bigmid \text{$i\in [t]$, $j\in [s]$,
 and $c\le (1/2+ \delta) \cdot d$}\right\}\subseteq W.
\]
(E3) implies that every $w_{b_c, j,i}\in W_1$ is not dominated by any vertex
in $D\cap \bigcup_{c\in [d]} B_c$. Therefore, it has to be dominated by or
included in $D\cap (C\cup W)$.

\medskip
If $|D\cap W_1|> (1/2- \delta)\cdot d $, then again we are done. So suppose
$|D\cap W_1|\le (1/2- \delta)\cdot d$. Without loss of generality let
\[
W_2:= \left\{w_{b_c, j,i} \Bigmid \text{$i\in [t]$, $j\in [s]$,
 and $c\le 2\delta d$}\right\}\subseteq W_1
\]
and assume $W_2 \cap D= \emptyset$. Thus $W_2$ has to be dominated by $D\cap
C$. For later purpose, let
\[
Y:= \big\{b_c\bigmid c\le 2\delta d\big\}.
\]
Obviously, $|Y|\ge 2\delta d -1$.

\medskip
Again we only need to consider the case $|D\cap C|\le (1/2- \delta)\cdot d$.
Recall $C= A(H)\times [t]$. Thus there is an $i \in [t]$ such that
\[
\Big|D\cap \big(A(H)\times \{i\}\big) \Big|
 \le \left(\frac{1}{2}- \delta\right)\cdot \frac{d}{t}.
\]
Let $X:= \big\{a\in A(H)\bigmid (a,i)\in D\big\}$, and in particular,
$|X|\le (1/2-\delta)\cdot d/ t$. Since $W_2$ is dominated by $D\cap C$, we
have for all $b\in Y$ and $j\in [s]$ there exists $a\in X$ such that
$\big\{(a, i), w_{b,j,i}\big\}\in E(G')$, which means that $\{a, b\}\in
E(H)$ and $\alpha(a)= j$. It follows that in the graph $H$ every vertex of
$Y$ has at least $s$ neighbors in $X$. Recall that $(1/2- \delta)\cdot d/
t\le \sqrt[2s]{d}$ by~\eqref{eq:at}. There are at most $\sqrt{d}$ different
types of $s$-vertex sets in $X$, i.e.,
\[
\left|\binom{X}{s}\right| \le \binom{(1/2-\delta)\cdot d/ t}{s}
 \le \left(\sqrt[2s]{d}\right)^{s} = \sqrt{d}.
\]
By the pigeonhole principle, there exists an $s$-vertex set of $X\subseteq
A(H)$ having at least $|Y|/ \sqrt{d}\ge  2\delta \sqrt{d}-1$ common
neighbors in $Y\subseteq B(H)$, which contradicts Claim~3. \hfill$\dashv$.

\medskip
Claim~2 and Claim~4 indeed imply that there is an fpt-reduction from the
clique problem to the dominating set problem which creates a gap great than
\[
\frac{3/2 - \delta}{1+\varepsilon}.
\]
So if there is a $\rho$-approximation of the dominating set problem,
by~\eqref{eq:rhoepsdelta} we can decide the clique problem in fpt time.
\proofend

\section{The Constant-Inapproximbility of \pds}\label{sec:main}

Theorem~\ref{thm:main1} is a fairly direct consequence of the following
theorem.
\begin{theo}[Main]\label{thm:mainreduction}
There is an algorithm $\mathbb A$ such that on input a graph $G$, $k\ge 3$,
and $c\in \mathbb N$ the algorithm $\mathbb A$ computes a graph $G_c$ such
that
\begin{enumerate}
\item[(i)] if $K_k\subseteq G$, then $\ds(G_c)< 1.1\cdot d^c$;

\item[(ii)] if $K_k\nsubseteq G$, then $\ds(G_c)> c\cdot d^c/3$,
\end{enumerate}
where $d= \left(30\cdot c^2\cdot (k+1)^2\right)^{4\cdot k^3+ 3c}$. Moreover
the running time of $\mathbb A$ is bounded by $f(k,c)\cdot |G|^{O(c)}$ for a
computable function $f: \mathbb N\times \mathbb N\to \mathbb N$.
%
%
\end{theo}

\begin{proof}[of Theorem~\ref{thm:main2}]
Suppose for some $\epsilon> 0$ there is an \fpt-algorithm $\mathbb{A}(G)$
which outputs a dominating set for $G$ of size at most
$\sqrt[4+\varepsilon]{\log (\ds(G))} \cdot \ds(G)$. Of course we can further
assume that $\varepsilon< 1$. Then on input a graph $G$ and $k\in\mathbb N$,
let
\begin{eqnarray*}
c := \ceil{k^{1-\epsilon/5}}= o(k)
 & \text{and} &
d:= \left(30\cdot c^2\cdot (k+1)^2\right)^{4\cdot k^3+3c}.
\end{eqnarray*}
We have
\[
\sqrt[4+\varepsilon]{\log(1.1\cdot d^c)}
= O\left(\sqrt[4+\varepsilon]{c\cdot k^3\cdot \log k}\right)
= o\left(k^{\frac{4}{4+\varepsilon}}\right)
= o(c).
\]
By Theorem~\ref{thm:mainreduction}, we can construct a graph $G_c$ with
properties (i) and (ii) in time
\[
f(k,c)\cdot |G|^{O(c)} = h(k)\cdot |G|^{o(k)}
\]
for an appropriate computable function $h:\mathbb N\to \mathbb N$. Thus, $G$
contains a clique of size $k$ if and only if $\mathbb A(G_c)$ returns a
dominating set of size at most
\[
1.1\cdot d^c\cdot \sqrt[4+\varepsilon]{\log(1.1\cdot d^c)}
 = o\left(c\cdot d^c \right)< \frac{c\cdot d^c}{3},
\]
where the inequality holds for sufficiently large $k$ \big(and hence
sufficiently large $c\cdot d^c$\big).

Therefore we can determine whether $G$ contains a $k$-clique in time
$g(k)\cdot |G|^{o(k)}$ for some computable $g:\mathbb N\to \mathbb N$. This
contradicts a result in Chen et.al.~\cite[Theorem~4.4]{chehua04} under \ETH.
\proofend
\end{proof}

\medskip
\subsection{Proof of Theorem~\ref{thm:mainreduction}}
We start by showing a variant of Theorem~\ref{thm:gapreduction}.

\begin{theo}\label{thm:colgap}
Let $\Delta\in \mathbb N^+$ be a constant and $d: \mathbb N^+\to \mathbb
N^+$ a computable function. Then there is an \fpt-algorithm that on input a
graph $G$ and a parameter $k\in \mathbb N$ with $6 \mid k+1$ constructs a
bipartite graph $H= \big(A(H) \dotcup B(H), E(H)\big)$ together with two
colorings
\begin{eqnarray*}
\alpha: A(H)\to [\Delta s] & \text{and} & \beta: B(H)\to [d(k)]
\end{eqnarray*}
such that:
\begin{enumerate}
\item[(H1)] if $K_k\subseteq G$, then there are $\Delta s$ vertices of
    distinct $\alpha$-colors in $A(H)$ with $d(k)$ common neighbors of
    distinct $\beta$-colors in $B(H)$;

\item[(H2)] if $K_k\nsubseteq G$, then every $\Delta(s-1)+ 1$ vertices in
    $A(H)$ have at most $(k+1)!$ common neighbors in $B(H)$,
\end{enumerate}
where $s=\binom{k}{2}$. 
\end{theo}


\proof Let $G$ be a graph with $n$ vertices and $k\in \mathbb N$. Assume
without loss of generality
\begin{eqnarray*}
\ceil{n^{\frac{6}{k+6}}}> (k+6)!
 & \text{and} &
\ceil {n^{\frac{6}{k+1}}} \ge d(k).
\end{eqnarray*}
By Theorem~\ref{thm:gapreduction} we can construct in polynomial time a
bipartite graph $H_0= (A_0 \dotcup B_0, E_0)$ such that for $s:=
\binom{k}{2}$:
\begin{itemize}
\item if $K_k\subseteq G$, then there are $s$ vertices in $A_0$ with at
    least $d(k)$ common neighbors in $B_0$;

\item if $K_k\nsubseteq G$, then every $s$ vertices in $A_0$ have at most
    $(k+1)!$ common neighbors in $B_0$.
\end{itemize}
Define
\begin{align*}
A_1 := A_0\times [\Delta], \ B_1:= B_0, \ \text{and}\
E_1 := \big\{\{(u,i), v\} \bigmid \text{$(u,i)\in A_0\times [\Delta]$,
  $v\in B_0$, and $\{u,v\}\in E_0$}\big\}.
\end{align*}
It is easy to verify that in the bipartite graph $(A_1\dotcup B_1, E_1)$
\begin{itemize}
\item if $K_k\subseteq G$, then there are $\Delta s$ vertices in $A_1$
    with at least $d(k)$ common neighbors in $B_2$;

\item if $K_k\nsubseteq G$, then every $\Delta (s-1)+1$ vertices in $A_1$
    have at most $(k+1)!$ common neighbors in $B_1$.
\end{itemize}
Applying Lemma~\ref{lem:cc} on
\begin{eqnarray*}
\big(n \gets |A_1|, k\gets \Delta s\big)
& \text{and} &
\big(n \gets |B_1|, k\gets d(k)\big)
\end{eqnarray*}
we obtain two function families $\Lambda_A:= \Lambda_{|A_1|, \Delta s}$ and
$\Lambda_B:= \Lambda_{|B_1|, d(k)}$ with the stated properties. Finally the
desired bipartite graph $H$ is defined by $\Big((A_1\times \Lambda_A\times
\Lambda_B) \dotcup (B_1\times \Lambda_A\times \Lambda_B), E)\Big)$
with
\[
E:= \Big\{\big\{(u,h_1,h_2), (v, h_1, h_2)\big\} \Bigmid
 \text{$u\in A_1$, $v\in B_1$, $h_1\in \Lambda_A$, $h_2\in \Lambda_B$,
   and $\{u,v\}\in E_1$}\Big\}
\]
and the colorings
\[
\alpha(u,h_1,h_2):= h_1(u)
 \quad \text{and} \quad
\beta(v,h_1, h_2):= h_2(v). \benda
\]

\bigskip
\noindent\textbf{Setting the parameters.} Let $\Delta:= 2$. Recall that
$k\ge 3$, $s=\binom{k}{2}\ge 3$, and $c\in \mathbb N^+$. We first
define
\[
d:= d(k):= \left(30\cdot c^2\cdot (k+1)^2\right)^{4\cdot k^3+3c}.
\]
It is easy to check that:
\begin{enumerate}
\item[(i)] $d^{\frac{1}{2}-\frac{1}{2s}}> c\cdot s^c \ \Big(\!= c \cdot
    \binom{k}{2}^c\Big)$.

\item[(ii)] $d> \big(3(k+1)!\big)^{2s}$.

\item[(iii)] $d> \left(10\Delta s\cdot c^2\right)^{2\Delta s}$.
\end{enumerate}
Then let
\begin{equation}\label{eq:t}
t:= c \cdot d^{c-\frac{1}{2\Delta s}}.\footnotemark
\end{equation}
\footnotetext{Here, we assume $d^{c-\frac{1}{2\Delta s}}$ is an integer.
Otherwise, let $d\gets d^{2\Delta s}$ which maintains (i)-- (iii).} From
(ii), (iii), and~\eqref{eq:t} we conclude
\begin{equation}\label{eq:dt}
\Delta s c t< 0.1\cdot d^c,
 \quad \frac{c\cdot d^c}{3t}\le \sqrt[2\Delta s]{d},
 \quad \text{and}\ (k+1)!< \frac{\sqrt[2 s]{d}}{3}.
\end{equation}
Moreover by (i) and $\Delta= 2$ we have
\begin{equation}\label{eq:deltas}
c\cdot d^c+ c\Delta^c s^c d^{c-\frac{1}{2}+\frac{1}{2s}}< 2\Delta^cd^c.
\end{equation}

\bigskip
\noindent\textbf{Construction of $G_c$.} We invoke Theorem~\ref{thm:colgap}
to obtain $H=(A \dotcup B, E)$, $\alpha$, and $\beta$. Then we construct a
new graph $G_c= \big(V(G_c), E(G_c)\big)$ as follows. First, the vertex set
of $G_c$ is given by
\[
V(G_c):= \bigcup_{\i\in [d]^c} V_{\i}\dotcup C\dotcup W,
\]
where
\begin{equation*}
V_{\i}:= \big\{\v\in B^{c}\bigmid \beta(\v)= \i \big\}
 \quad \text{for every $\i\in [d]^c$},
\end{equation*}
\vspace{-.7cm}
\begin{eqnarray*}
C:= A \times [c]\times [t], & \text{and} &
W:= \Big\{w_{\v,\j,i}\Bigmid \text{$\v\in V_{\i}$ for some $\i\in [d]^{c}$,
 $\j\in [\Delta s]^c$ and $i\in [t]$}\Big\}.
\end{eqnarray*}
Moreover, $G_c$ contains the following types of edges.
\begin{enumerate}
\item[(E1)] For each $\i\in [d]^{c}$, $V_{\i}$ forms a clique.


\item[(E2)] Let $\i\in [d]^{c}$ and $\v, \v'\in V_{\i}$. If for all
    $\ell\in [c]$ we have $\v(\ell)\ne \v'(\ell)$ then $\left\{w_{\v, \j,
    i}, \v'\right\}\in E(G_c)$ for every $i\in [t]$ and\; $\j\in[\Delta
    s]^c$.

\item[(E3)] Let $i\in [t]$. Then $\big\{(u,\ell,i), w_{\v, \j, i}\big\}\in
    E(G_c)$ if $\{u, \v(\ell)\}\in E$ and\;
    $\j(\ell)= \alpha(u)$.

\item[(E4)] Let $u, u'\in A(H)$ with $u\ne u'$, $\ell\in[c]$, and $i\in
    [t]$. Then $\big\{(u,\ell,i), (u',\ell,i)\big\}\in E(G_c)$.
\end{enumerate}

Theorem~\ref{thm:mainreduction} then follows from the completeness and the
 soundness of this reduction.

\begin{lem}[Completeness]\label{lem:completness}
If $G$ contains $k$-clique, then $\ds(G_c)< 1.1
d^c$.
\end{lem}

\begin{lem}[Soundness]\label{lem:soundness}
If $G$ contains no $k$-clique then $\ds(G_c)>c\cdot d^c/3$.
\end{lem}


\medskip
We first show the easier completeness.

\begin{proof}[of Lemma~\ref{lem:completness}]
By (H1) in Theorem~\ref{thm:colgap}, if $G$ contains a subgraph isomorphic
to $K_k$, then the bipartite graph $H$ has a $K_{\Delta s,d}$-subgraph $K$
such that $\alpha(A\cap K)= [\Delta s]$ and $\beta(B\cap K)= [d]$. Let
\[
D:= (B\cap K)^c\dotcup \big((A\cap K)\times[c]\times [t]\big).
\]
Obviously, $|D|= d^c+ \Delta sct< 1.1\cdot d^c$ by~\eqref{eq:dt}. And (E1)
and (E4) imply that $D$ dominates every vertex in $C$ and every vertex in
$V_{\i}$ for all $\i\in [d]^c$.

To see that $D$ also dominates $W$, let $w_{\v,\j,i}$ be a vertex in $W$.
First consider the case where $\v(\ell)\notin B\cap K$ for all $\ell\in[c]$.
Since $\beta\big((B\cap K)^c\big)= [d]^c$, there exists a vertex $\v'\in
(B\cap K)^c$ with $\beta(\v')= \beta(\v)$ and $\v(\ell)\neq\v'(\ell)$ for
all $\ell\in[c]$. Then $w_{\v, \j, i}$ is dominated by $\v'$ because of
(E2).

Otherwise assume $\v(\ell)\in B\cap K$ for some $\ell\in[c]$, then $A\cap
K\subseteq N^H(\v(\ell))= \big\{u\in A\bigmid \{u, \v(\ell)\}\in
E\big\}$.
There exists a vertex $u\in A\cap K$ such that $\alpha(u)= \j(\ell)$ and
$\big\{\v(\ell),u\big\}\in E$. By (E3), $w_{\v,\j,
i}$ is adjacent to $(u,\ell,i)$. \proofend
\end{proof}

\medskip
\subsection{Soundness}

\begin{lem}\label{lem:productgap}
Suppose $c, \Delta, t\in\mathbb{N}^+$ and $\Delta< t$. Let $V\subseteq
[t]^c$. If there exists a function $\theta: V\to [c]$ such that for all
$i\in [c]$ we have
\begin{equation}\label{eq:V}
\Big|\big\{\v(i)\bigmid \text{$\v\in V$ and $\theta(\v)= i$}\big\}\Big|
 \le t- \Delta,
\end{equation}
then $|V|\le t^c- \Delta^c$.
\end{lem}

\begin{proof}
When $c=1$, we have $|V|\le t- \Delta$ by~\eqref{eq:V}. Suppose the lemma
holds for $c\le n$ and consider $c= n+1$. Given $V\subseteq [t]^{n+1}$ and
$\theta$, let
\[
C_{n+1}:= \big\{\v(n+1)\bigmid \text{$\v\in V$ and $\theta(\v)= n+1$} \big\}.
\]
By~\eqref{eq:V}, $|C_{n+1}|\le t- \Delta$. If $|C_{n+1}|< t- \Delta$, we add
$\big(t-\Delta-|C_{n+1}|\big)$ arbitrary integers from $[t]\setminus
C_{n+1}$ to $C_{n+1}$. So we have $|C_{n+1}|= t- \Delta$. Let $A:=
\big\{\v\in V\bigmid \v(n+1)\in C_{n+1}\big\}$ and $B:= V\setminus A$. It
follows that
\begin{equation}\label{eq:A}
|A|\le (t- \Delta)t^{c-1},
\end{equation}
$\Big|\big\{\v(n+1)\bigmid \v\in B\big\}\Big|\le \Delta$, and
$\theta(\v)\in[c-1]$ for $\v\in B$. Let
\[
V':= \big\{(v_1,v_2,\cdots,v_n)\bigmid
 \exists v_{n+1}\in [t], (v_1,v_2,\cdots,v_n,v_{n+1})\in B\big\}.
\]
We define a function $\theta': V'\to [c-1]$ as follows. For all $\v'\in V'$,
choose $\v\in B$ with the minimum $\v(c)$ such that for all $i\in[c-1]$ it
holds $\v'(i)= \v(i)$. By the definition of $V'$, such a $\v$ must exist,
and we let $\theta'(\v'):= \theta(\v)$.  By~(\ref{eq:V}),
$\Big|\big\{\v'(i)\bigmid \text{$\v'\in V'$ and $\theta'(\v')=
i$}\big\}\Big|\le t- \Delta$ for all $i\in [c-1]$. Applying the induction
hypothesis, we get $|V'|\le t^{c-1}- \Delta^{c-1}$. Obviously,
\begin{equation}\label{eq:B}
|B|\le \Delta|V'|\le \Delta t^{c-1}- \Delta^c.
\end{equation}
From~\eqref{eq:A} and~\eqref{eq:B}, we deduce that $|V|= |A|+ |B|\le
(t-\Delta)t^{c-1}+ \Delta t^{c-1}- \Delta^c \le t^c- \Delta^c$. \proofend
\end{proof}


We are now ready to prove the soundness of our reduction.

\begin{proof}[of Lemma~\ref{lem:soundness}]
Let $D$ be a dominating set of $G_c$.
Define
\[
a:= \Big|\big\{\i\in [d]^c \bigmid |D\cap V_{\i}|\ge c+1 \big\}\Big|.
\]
If $a> d^c/3$, then $|D|\ge (c+1)a> c\cdot d^c/3$ and we are done.

\medskip
So let us consider $a\le d^c/3$. Thus, the set
\[
I:= \big\{\i\in [d]^c \bigmid |D\cap V_{\i}|\le c \big\}
\]
has size $|I|\ge 2d^c/3$. Let $\i\in I$ and assume that $D\cap V_{\i}=
\big\{\v_1, \v_2, \ldots, \v_{c'}\big\}$ for some $c'\le c$. We define a
$\v_\i\in V_\i$ as follows. If $c'=0$, we choose an arbitrary $\v_\i\in
V_\i$.\footnote{Since the coloring $\beta$ is obtained by the color-coding
used in the proof of Theorem~\ref{thm:colgap}, for every $b\in[d]$ it holds
that $\{v\in B\mid \beta(v)=b\}\neq\emptyset$, hence $V_\i\neq
\emptyset$.} Otherwise, let
\[
\v_\i(\ell):=
 \begin{cases}
  \v_\ell(\ell) & \text{for all $\ell\in[c']$}; \\
  \v_1(\ell) & \text{for all $c'<\ell\le c$}. \\
 \end{cases}
\]
Obviously, $\beta(\v_\i)= \i$.

\smallskip
(E2) implies that for every $\j\in [\Delta s]^c$ and every $i\in [t]$, the
vertex $w_{\v_{\i},\j,i}$ is not dominated by $D\cap V_{\i}$. Observe that
$w_{\v_{\i},\j,i}$ cannot be dominated by other $D\cap V_{\i'}$ with $\i'\ne
\i$ either, by (E2) and (E3). Therefore every vertex
in the set
\[
W_1:= \big\{w_{\v_{\i}, \j,i}
  \bigmid \text{$\i\in I$, $\j\in [\Delta s]^c$, and $i\in [t]$} \big\}
\]
is \emph{not} dominated by $D\cap \bigcup_{\i\in [d]^c} V_{\i}$. As a
consequence, $W_1$ has to be dominated by or included in $D\cap (C\cup W)$.

\medskip
If $|D\cap W_1|> c\cdot d^c/3$, then again we are done. So suppose $|D\cap
W_1|\le c\cdot d^c/3$ and let $W_2:= W_1\setminus D$. It follows that $W_2$
has to be dominated by $D\cap C$. Once again we only need to consider the
case $|D\cap C|\le c\cdot d^c/3$, and hence there is an $i' \in [t]$ such
that
\begin{equation}\label{eq:DCsize}
\Big|D\cap \big(A\times[c]\times \{i'\}\big) \Big|
 \le \frac{c\cdot d^c}{3 t}.
\end{equation}
Then we define
\[
Z:= \big\{w_{\v,\j, i}\in W_2\bigmid i=i'\big\}
 = \big\{w_{\v_{\i},\j, i'}\bigmid \text{$\i\in I$,
     $\j\in[\Delta s]^c$, and $w_{\v_{\i},\j, i'}\notin D$}\big\}.
\]
So $Z$ has to be dominated by $D\cap C$, and in particular those vertices of
the form $(u,\ell,i')\in D\cap C$. Moreover,
\begin{equation}\label{eq:Z}
|Z|\ge \Delta^cs^c|I|- |D\cap W_1|\ge \Delta^cs^c |I|- c\cdot d^c/3.
\end{equation}

Our next step is to upper bound $|Z|$. To that end, let
\[
X:= \big\{u\in A\bigmid \text{$(u,\ell,i')\in D$ for some $\ell\in [c]$}\big\}.
\]
Thus $Z$ is dominated by those vertices $(u,\ell,i')$ with $u\in X$. And
by~\eqref{eq:DCsize}
\[
|X|\le \frac{c\cdot d^c}{3 t}.
\]
Set
\[
Y:= \Big\{v\in B\Bigmid \big|N^H(v)\cap X\big|> \Delta(s-1)\Big\}.
\]
Recall that $c\cdot d^c/(3t)\le \sqrt[2\Delta s]{d}$ by~\eqref{eq:dt}. Hence
$X$ has at most $\sqrt{d}$ different subsets of size $\Delta(s-1)+ 1$, i.e.,
\[
\left|\binom{X}{\Delta(s-1)+1}\right|\le |X|^{\Delta(s-1)+1}\le |X|^{\Delta s}
\le \sqrt{d}.
\]
We should have
\begin{equation}\label{eq:Y}
|Y|\le \sqrt{d}\cdot (k+1)!\le \frac{d^{\frac{1}{2}+\frac{1}{2s}}}{3},
\end{equation}
where the second inequality is by~\eqref{eq:dt}.
 Otherwise, by the pigeonhole
principle, there exists a $(\Delta(s-1)+1)$-vertex set of $X\subseteq A(H)$
having at least $|Y|/\sqrt{d}> (k+1)!$
common neighbors in $Y\subseteq B(H)$. However, if $G$ contains no
$k$-clique, then by (H2) every $\big(\Delta(s-1)+1\big)$-vertex set of
$A(H)$ has at most $(k+1)!$ common neighbors in $B(H)$, and we obtain a
contradiction.

\medskip
Let
\begin{align*}
Z_1 := & \big\{w_{\v, \j, i'}\in Z\bigmid \text{there exists an $\ell\in[c]$ with $\v(\ell)\in Y$}\big\}
  \quad\Big(\!\!\subseteq Z\Big) \\
   = & \big\{w_{\v_{\i}, \j, i'}\bigmid \text{$\i\in I$, $\j\in [\Delta s]^c$, $w_{\v_{\i}, \j, i'}\notin D$,
   and there exists an $\ell\in[c]$ with $\v_\i(\ell)\in Y$}\big\} \\
 \text{and}\qquad
Z_2 := & Z\setminus Z_1
 = \big\{w_{\v_{\i}, \j, i'}\bigmid \text{$\i\in I$, $\j\in [\Delta s]^c$, $w_{\v_{\i}, \j, i'}\notin D$,
   and $\v_{\i}(\ell)\notin Y$ for all $\ell\in [c]$}\big\}.
\end{align*}
Moreover, let $I_1:=\{\i\in I\mid \text{there exists a $w_{\v_\i, \j, i'}\in
Z_1$}\}$. From the definition, we can deduce that
\[
\text{for all $\i\in I_1$ there exists an $\ell\in [c]$ such that $\i(\ell)\in \beta(Y)$}.
\]
Then $|I_1|\le c|Y|d^{c-1}$ and hence
\[
|Z_1|\le|I_1|\Delta^cs^c\le c |Y| d^{c-1} \Delta^c s^c.
\]

To estimate $|Z_2|$, let us fix an $\i\in I$ and thus fix the tuple
$\v_{\i}\in B^c$, and consider the set
\[
J_{\i}:= \big\{\;\j\in[\Delta s]^c\bigmid w_{\v_{\i},\j, i'}\in Z_2 \big\}.
\]
Recall that $Z$ is dominated by those vertices $(u, \ell, i')$ with $u\in
X$, so for every $\j\in J_\i$ the vertex $w_{\v_i, \j, i'}$ is adjacent to
some $(u, \ell, i')$ in the dominating set $D$ with $u\in X$. Moreover, for
every $\ell\in [c]$, in the original graph $H$ the vertex $\v_{\i}(\ell)\in
B$ has at most $\Delta(s-1)$ neighbors in $X$, by the fact that
$\v_{\i}(\ell)\notin Y$ and our definition of the set $Y$.

Define a function $\theta: J_{\i}\to [\Delta s]$ such that for each\; $\j\in
J_{\i}$, if $w_{\v_{\i}, \j, i'}$ is adjacent to a vertex $(u, \ell, i')\in
D$ with $u\in X$, then $\theta(\j)= \ell$. As argued above, such a $(u,
\ell, i')$ must exist, and if there are more than one such, choose an
arbitrary one.

Let \;$\j\in J_{\i}$ and $\ell:= \theta(\j)$. By (E3), in the graph $H$ the
vertex $\v_{\i}(\ell)$ is adjacent to some vertex $u\in X$ with $\alpha(u)=
\j(\ell)$. It follows that for each $\ell\in[c]$ we have
\[
\Big|\big\{\;\j(\ell)\bigmid \text{$\j\in J_{\i}$ and $\theta(\j)= \ell$}\big\}\Big|
\le \Big|\big\{\alpha(u)\bigmid \text{$u\in X$ adjacent to $\v_{\i}(\ell)$}\big\}\Big|
\le \Delta(s-1).
\]
Applying Lemma~\ref{lem:productgap}, we obtain
\begin{equation*}
\big|J_{\i}\big|\le \Delta^cs^c- \Delta^c.
\end{equation*}
Then
\[
\big|Z_2\big|= \sum_{\i\in I} \big|J_{\i}\big|
 \le |I|(\Delta^cs^c-\Delta^c).
\]
By~\eqref{eq:Z} and the definition of $Z_1$ and $Z_2$, we should have
\[
\Delta^cs^c|I|- c\cdot d^c/3\le |Z|=|Z_1|+|Z_2|
 \le c |Y| d^{c-1} \Delta^c s^c + |I|(\Delta^cs^c-\Delta^c).
\]
That is,
\[
c\cdot d^c/3 + c |Y| d^{c-1} \Delta^c s^c\ge \Delta^c |I|\ge 2\Delta^cd^c/3.
\]
Combined with~\eqref{eq:Y}, we have
\[
c\cdot d^c +c \Delta^c s^cd^{c-\frac{1}{2}+\frac{1}{2s}}\ge 2\Delta^cd^c,
\]
which contradicts the equation~\eqref{eq:deltas}. \proofend
\end{proof}

\section{Some Consequences}\label{sec:consq}


\begin{proof}[of Corollary~\ref{cor:main1}]
Let $c\in \mathbb N^+$, and assume that $\mathbb A$ is a polynomial time
algorithm which on input a graph $G= (V,E)$ with $\ds(G)\le \beta(|V|)$
outputs a dominating set $D$ with $|D|\le c\cdot \ds(G)$. Without loss of
generality, we further assume that given $0\le k\le n$ it can be tested in
time $n^{O(1)}$ whether $k> c\cdot \beta(n)$.

Now let $G$ be an arbitrary graph. We first simulate $\mathbb A$ on $G$, and
there are three possible outcomes of $\mathbb A$.
\begin{itemize}
\item $\mathbb A$ does not output a dominating set. Then we know $\ds(G)
    > \beta(|V|)$. So in time
    \[
    2^{O(|V|)} \le 2^{O(\beta^{-1}(\ds(G)))}
    \]
    we can exhaustively search for a minimum dominating set $D$ of $G$.

\item $\mathbb A$ outputs a dominating set $D_0$ with $|D_0|> c\cdot
    \beta(|V|)$. We claim that again $\ds(G)> \beta(|V|)$. Otherwise, the
    algorithm $\mathbb A$ would have behaved correctly with
    \[
    |D_0|\le c\cdot \ds(G)\le c\cdot \beta(|V|).
    \]
    So we do the same brute-force search as above.

\item $\mathbb A$ outputs a dominating set $D_0$ with $|D_0|\le c\cdot
    \beta(|V|)$. If $|D_0|> c\cdot \ds(G)$, then
    \[
    c\cdot \beta(|V|) \ge |D_0|> c\cdot \ds(G), \quad \text{i.e.},
     \ \beta(|V|)> \ds(G),
    \]
    which contradicts our assumption for $\mathbb A$. Hence, $|D_0|\le
    c\cdot \ds(G)$ and we can output $D:= D_0$.
\end{itemize}
To summarize, we can compute a dominating set $D$ with $|D|\le c\cdot
\ds(G)$ in time $f(\ds(G))\cdot |G|^{O(1)}$ for some computable $f:\mathbb
N\to \mathbb N$. This is a contradiction to Theorem~\ref{thm:main1}.
\proofend
\end{proof}

\medskip
Now we come to the approximability of the monotone circuit satisfiability
problem.
\noptprob[6]{$\mcs$}{A monotone circuit $C$}{A satisfying assignment $S$ of
$C$}{The weight of $|S|$}{$\min$}
Recall that a Boolean circuit $C$ is \emph{monotone} if it contains no
negation gates; and the \emph{weight} of an assignment is the number of
inputs assigned to $1$.

As mentioned in the Introduction, Marx showed~\cite{mar13} that \mcs\ has no
fpt approximation with any ratio $\rho$ for circuits of depth 4, unless
$\FPT= \W 2$.

\begin{cor}\label{cor:wsat}
Assume $\FPT\ne \W 1$. Then \mcs\ has no constant fpt approximation for
circuits of depth 2.
\end{cor}

\begin{proof}
This is an immediate consequence of Theorem~\ref{thm:main1} and the
following well-known approximation-preserving reduction from \mcs\ to \mds.
Let $G= (V,E)$ be a graph. We define a circuit
\[
C(G)= \bigwedge_{v\in V} \bigvee_{\{u,v\}\in E} X_u.
\]
There is a one-one correspondence between a dominating set in $G$ of size
$k$ and a satisfying assignment of $C(G)$ of weight $k$. \proofend
\end{proof}

\begin{rem}
Of course the constant ratio in Corollary~\ref{cor:wsat} can be improved
according to Theorem~\ref{thm:main2}.
\end{rem}

\section{Conclusions}\label{sec:con}

We have shown that $\pds$ has \emph{no} fpt approximation with any constant
ratio, and in fact with a ratio slightly super-constant. The immediate
question is whether the problem has fpt approximation with \emph{some} ratio
$\rho:\mathbb N\to \mathbb N$, e.g., $\rho(k)= 2^{2^k}$. We tend to believe
that it is not the case.

Our proof does not rely on the deep PCP theorem, instead it exploits the gap
created in the \W 1-hardness proof of the parameterized biclique problem
in~\cite{lin15}. In the same paper, the second author has already proved
some inapproximability result which was shown by the PCP theorem before.
Except for the derandomization using algebraic geometry in~\cite{lin15} the
proofs are mostly elementary. Of course we are working under some stronger
assumptions, i.e., \ETH\ and $\FPT\ne \W 1$. It remains to be seen whether
we can take full advantage of such assumptions to prove lower bounds
matching those classical ones or even improve them as in
Corollary~\ref{cor:main1}.

\medskip
\subsection*{Acknowledgement}
We thank Edouard Bonnet for pointing out a mistake in an earlier version of the paper.

}
\bibliographystyle{plain}
\bibliography{appds}
\end{document}